\documentclass[11pt,english]{article}

\pdfoutput=1

\usepackage[utf8]{inputenc}
\usepackage{amsmath}
\usepackage{ltexpprt} 
\usepackage{makeidx} 
\usepackage{subfigure}
\usepackage{multirow}
\usepackage{graphicx}
\usepackage{epstopdf}
\usepackage{txfonts}
\usepackage{epsfig}
\usepackage{wrapfig}
\usepackage{url}
\usepackage{stfloats}
\usepackage{enumitem}
\usepackage{verbatim}
\setlist[itemize]{noitemsep}
\usepackage{clrscode3e}
\usepackage{listings}
\usepackage{multirow}
\usepackage{color,soul}
\usepackage{colortbl}
\usepackage{array}
\usepackage{tikz}
\usepackage{caption}
\usepackage[normalem]{ulem}

\usepackage{color,soul}
\usepackage{colortbl}

\def\knn{{$k$NN}}
\def\kNN{{$k$NN}}
\def\RkNN{{R$k$NN}}
\def\rknn{{R$k$NN}}
\def\rehub{{\emph{ReHub}}}
\def\ReHub{{\emph{ReHub}}}
\def\Rehub{{\emph{ReHub}}}

\DeclareCaptionLabelFormat{andtable}{#1~#2  \&  \tablename~\thetable}

\newcolumntype{M}[1]{>{\centering\arraybackslash}m{#1}}

\newcommand\tikzmark[2]{%
\tikz[remember picture,overlay] 
\node[inner sep=0pt,outer sep=2pt] (#1){#2};%
}

\newcommand\link[2]{%
\begin{tikzpicture}[remember picture, overlay, >=stealth, shorten >= 1pt]
  \draw[thick, red, ->] (#1.east) to  (#2.west);
\end{tikzpicture}%
}

\newcommand\linkkk[2]{%
\begin{tikzpicture}[remember picture, overlay, >=stealth, shorten >= 1pt]
  \draw[thick, blue, ->] (#1.east) to  (#2.west);
\end{tikzpicture}%
}

\newcommand\linkk[2]{%
\begin{tikzpicture}[remember picture, overlay, >=stealth, shorten >= 0pt]
  \draw [very thick,  dashed, ->] (#1.east) to  (#2.west);
\end{tikzpicture}%
}

\usepackage{palatino}

\def\knn{{$k$NN}}
\def\kNN{{$k$NN}}
\def\RkNN{{R$k$NN}}
\def\rknn{{R$k$NN}}
\def\rehub{{\emph{ReHub}}}
\def\ReHub{{\emph{ReHub}}}
\def\Rehub{{\emph{ReHub}}}

\pdfoutput=1
\date{\today}
\begin{document}
\title{ReHub. Extending Hub Labels for Reverse $k$-Nearest Neighbor Queries on Large-Scale networks}


\author{
  Alexandros Efentakis\\  \small{Research Center ``Athena''} 
 \\ \small{efentakis@imis.athena-innovation.gr}
  \and 
 Dieter Pfoser\\ \small{George Mason University} 
 \\ \small{dpfoser@gmu.edu}
}



\maketitle

\begin{abstract}
Quite recently, the algorithmic community has focused on solving multiple shortest-path query problems beyond simple vertex-to-vertex queries, especially in the context of road networks. Unfortunately, this research cannot be generalized for large-scale graphs, e.g., social or collaboration networks, or to efficiently answer Reverse $k$-Nearest Neighbor (\RkNN) queries, which are of practical relevance to a wide range of applications. To remedy this, we propose \ReHub, a novel main-memory algorithm that extends the Hub Labeling technique to efficiently answer \RkNN\ queries on large-scale networks. Our experimentation will show that \ReHub\ is the best overall solution for this type of queries, requiring only minimal preprocessing and providing very fast query times.
\end{abstract}
%


\section{Introduction}
\label{sec:intro} 

During the last two decades, the algorithmic community has produced significant results regarding vertex-to-vertex shortest-path queries, especially in the context of transportation networks (cf.~\cite{bast2014} for the latest overview). Recently, this focus shifted to additional types of shortest-path (SP) queries, 
such as \emph{one-to-all} (finding SP distances from a source vertex $s$ to all other graph vertices), \emph{one-to-many} (computing the SP distances between the source vertex $s$ and all vertices of a set of targets~$T$), \emph{range} (find all nodes reachable from $s$ within a given timespan), \emph{many-to-many} (calculate a distance table between two sets of vertices $S$ and $T$) and \kNN\ queries. Recent contributions here include \cite{delling2011phast} (one-to-all), \cite{delling2011g} (one-to-many, many-to-many), \cite{efentakis2014} (one-to-all, range, one-to-many) and  \cite{delling2013h,efentakis2014c} ($k$NN queries). Unfortunately, most of these methods target road networks and, thus, cannot easily be used for denser, small-diameter graphs, such as social and collaboration networks. 

In the case of large-scale networks, the prevailing technique for vertex-to-vertex shortest-path queries is based on the 2-hop labeling, or, \emph{Hub Labeling} (HL) algorithm \cite{gavoille2001,cohen2002}. During preprocessing, we calculate for every vertex $v$ a forward label $L_f(v)$ and a backward label~$L_b(v)$. These labels are subsequently used to very fast answer vertex-to-vertex shortest-path queries. 
The HL technique has been adapted successfully to road networks \cite{abraham2011f,abraham2012f,delling2013v,akiba2014f} and quite recently has also been extended to  undirected, unweighted graphs~\cite{akiba2013f,delling2014f,jiang2014}. 
The HL method has also been used for one-to-many, many-to-many and \kNN\ queries in road networks~\cite{delling2011g,abraham2012hldb}.

Another very important type of queries is the \emph{Reverse $k$-nearest neighbor} (\RkNN) problem, initially proposed in \cite{korn2000}. Given a query point~$q$ and a set of objects $P$, the \RkNN\ query retrieves all the objects in $P$ that have $q$ as one of their $k$-nearest neighbors according to a distance function $dist()$. 
In Euclidean space, the distance $dist(s,t)$ refers to the Euclidean distance between  two objects $s$ and~$t$. 
For graphs, $dist(s,t)$ corresponds to the minimum network distance between the two objects. \RkNN\ queries may be used in various domains, ranging from geomarketing to location-based services and a wide-range of applications, including resource allocation, profile-based marketing and decision support~\cite{encyclopedia2009}. Despite their importance and the fact that there is some scientific literature discussing \RkNN\ queries for road networks \cite{safar2009,cheema2012,borutta2014}, to the best of our knowledge, the only \RkNN\ work focusing on other types of graphs is~\cite{yiu2006}. Unfortunately, all those previous works share some inherent limitations, such as assuming that the graph does not fit in main memory (and therefore is stored on secondary storage), require query times of a few seconds which prohibits their use in real-time applications and most importantly, they do not scale particularly well with respect to the network size, the object density, the distribution of objects and the cardinality of the reverse $k$-nearest neighbor result.  

Putting everything together, the ambition of this work is to provide an efficient and fast main-memory algorithm for answering \RkNN\ queries on large-scale graphs. 
Our proposed algorithm, termed \ReHub\ (\emph{Re}verse \kNN\ + \emph{Hub} labels) extends the Hub Labeling approach to efficiently handle these queries. 
The main advantage of \ReHub\ is that its slower \emph{Offline phase} depends only on the location of the objects $P$ and has to run only once, whereas its \emph{Online phase} (which depends on the query vertex $q$) is very fast. Still, even the costlier offline phase hardly needs more than~1$s$, while the online phase requires typically less than~$1ms$, making \ReHub\ the only \RkNN\ algorithm fast enough for real-time applications and big, large-scale graphs. Moreover, the necessary data structures for answering \RkNN\ queries may also answer \kNN\ queries and require only a small fraction of the memory required for storing the created hub labels for the typical case of vertex-to-vertex queries. 
Throughout this work, we use undirected and unweighted graphs which constitute an important graph class (containing social and collaboration networks) but also pose a significant challenge to Hub Labeling algorithms because of the sheer size of the created labels. However, our method could be easily adapted for other graph classes where the hub-labeling algorithm typically performs well, including road networks.

The outline of this work is as follows. Section~\ref{sec:related_work} presents related work. Section~\ref{sec:contribution} describes the \ReHub\ algorithm and provides a theoretical analysis of its performance. Experiments showcasing \ReHub 's benefits are provided in Section~\ref{sec:experiments}. Finally, Section~\ref{sec:conclusions} gives conclusions and directions for future work.

\section{Related work}
\label{sec:related_work}

Given a query point $q$ and a set of objects $P$, the \emph{RkNN} query (also referred as the monochromatic \RkNN\ query) retrieves all the objects that have $q$ as one of their $k$-nearest neighbors, according to a distance function $dist()$. Formally \RkNN$(q)~=~\{p \in~P:dist(p,q) \leq dist(p,p_k) \}$ where $p_k$ is the $k$-Nearest Neighbor (\knn) of $p$. In Euclidean space, the distance $dist(s,t)$ refers to the Euclidean distance between two objects $s$ and~$t$. In graph networks, $dist(s,t)$ corresponds to the minimum network distance between the two objects. Throughout this work we use undirected, unweighted graphs $G(V,E)$ (where $V$ represent vertices and $E$ represents arcs), we assume that objects are located on vertices and we refer to \emph{snapshot} \RkNN\ queries, i.e, objects are not moving. Also, similar to previous works, the term \emph{object density} $D$ refers to the ratio~$|P|/|V|$. 

There is extensive literature focusing on \RkNN\ queries in Euclidean space. Since our work focuses on graphs, we only discuss the latter. 
Regarding road networks, the work of~\cite{safar2009} uses Network Voronoi cells (i.e., the set of vertices and arcs that are closer to the generator object) to answer \RkNN\ queries. This work has only been tested on a relatively small network ($110K$ arcs) and all precomputed information is stored in a database. Despite its costly preprocessing (for calculating the Network Voronoi cells), queries still require $1.5s$ for $D=0.05$ and $k=1$. The query times further increase to $32s$ for $k=20$ . 
Later works focusing on \emph{continuous \RkNN\ queries} on road networks~\cite{cheema2012} have only been tested with even smaller road networks ($22K$ arcs) and are different in scope from our work, which focuses on snapshot \RkNN\ queries. 
To the best of our knowledge, the only work focusing on other graph classes (besides road networks) is~\cite{yiu2006}. This work, too, has only been tested on sparse networks, e.g., road networks, grid networks (max degree 10), p2p graphs (avg degree 4) and a very small, sparse co-authorship graph ($4K$ nodes). Furthermore, all experimentation there for values of~$k>1$ (up to $k=8$) refers to road networks, so the scalability of the proposed algorithms for denser graphs and larger values of $k$ is debatable. Interestingly enough, this work  proposed the \emph{Eager M} algorithm that, similar to \ReHub, has an offline and an online phase (that uses the precomputed information obtained from the offline phase) to accelerate \RkNN\ queries. Unfortunately, both phases are unoptimized. The offline phase uses a slow, combined network expansion from all objects, which cannot scale for very dense graphs or sparse objects. The faster online phase performs a pruned Dijkstra-like expansion from the query vertex and thus, will also be too slow for denser graphs and small values of $D$. 
Recently, Borutta et al.~\cite{borutta2014} extended this work for time-dependent road networks, but the presented results were also not encouraging. The larger road network tested had $50k$ vertices (queries require more than $1s$ for $k=1$) and for a road network of $10k$ nodes and $k=8$, \RkNN\ queries take more than $0.3s$ (without even adding the I/O cost). 
In a nutshell, all existing contributions and methods have not been tested on dense, large-scale graphs, do not scale well for increasing $k$ values, and their performance highly depends on the object density $D$.       

Our work builds on the 2-hop labeling or Hub Labeling (HL) algorithm of \cite{gavoille2001,cohen2002} in which, during preprocessing, we store at every vertex~$v$ a forward $L_f(v)$ and a backward label $L_b(v)$. 
The forward label $L_f(v)$ is a sequence of pairs $(u, dist(v, u))$, with $u \in V$. 
Likewise, the backward label $L_b(v)$ contains pairs $(w, dist(w, v))$.
Vertices $u$ and $w$ denote the \emph{hubs of v}. 
The generated labels conform to \emph{the cover property}, i.e., for any $s$ and $t$, the set $L_f(s) \cap L_b(t)$ must contain at least one hub that is on the shortest $s-t$ path. 
For undirected graphs $L_b(v) = L_f(v)$. To find the network distance $dist(s, t)$ between two vertices $s$ and $t$, a HL query must find the hub $v \in L_f(s) \cap L_b(t)$ that minimizes the sum $dist(s, v)+dist(v,t)$. Since the pairs in each label are sorted by hub, this takes linear time by employing a coordinated sweep over both labels. The HL technique has been successfully used for road networks in \cite{abraham2011f,abraham2012f,delling2013v,akiba2014f}. In the case of large-scale graphs, the Pruned Landmark Labeling (PLL) algorithm of~\cite{akiba2013f} ``\emph{produces a minimal labeling for a specified vertex ordering}''~\cite{delling2014f}. 
Although this work orders vertices by degree, the later work of~\cite{delling2014f} improves the suggested vertex ordering and the storage schema of the hub labels for maximum compression. On a similar note, Jiang et al.~\cite{jiang2014} propose their HopDB algorithm to provide an efficient HL index construction when the given graphs and the corresponding index are too big to fit into main memory. The HL method has also been used for one-to-many, many-to-many and \kNN\ queries on road networks in~\cite{delling2011g} and~\cite{abraham2012hldb} respectively. 
The core contribution of our work is to extend existing HL techniques in the context of \RkNN\ queries on large-scale graphs and the proposed \Rehub\ algorithm, presented in the following section.



\section{The \emph{ReHub} algorithm}
\label{sec:contribution} 

What follows is the description of the \ReHub\ (\emph{Re}verse \kNN + \emph{Hub} labels) algorithm that extends the Hub Labeling approach to efficiently handle \RkNN\ queries on large-scale graphs. 
\ReHub\ consists of two distinct, independent phases:  
(i)~A slower, costlier \emph{Offline} phase that takes place after the creation of the hub labels and depends only on the objects $P$ (regardless of the query vertex $q$).
(ii)~An \emph{Online} phase that uses the auxiliary data structures created during the Offline phase to compute the actual \rknn\ query results. 
The main benefit of the \ReHub\ algorithm is that the costlier offline phase has to run only once and may service all \RkNN\ queries for a specific set of objects, whereas the online phase (that actually depends on the query vertex $q$) is very fast (typically less than a $1ms$). 
Hence, \ReHub\ may be used within the context of real-time applications, operating on large-scale graphs.

\subsection{Offline Phase}
\label{sub:offline} 

 \begin{figure}[!tb]
 \centering
 \includegraphics[width=0.21\columnwidth]{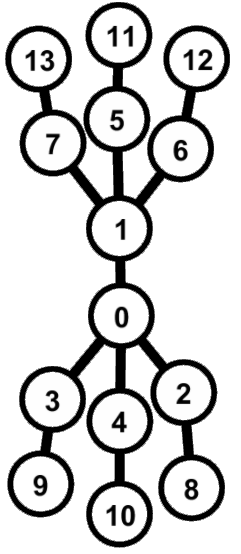}
 \qquad
 \begin{tabular}[b]{|c|c|}  \hline 
\bfseries{Vertex} & \bfseries{Hub Labels (h,d)} \\ \hline 
0 & (0,0)\\ \hline
1 & (0,1), (1,0)\\ \hline
2 & (0,1), (2,0)\\ \hline
3 & (0,1), (3,0)\\ \hline
\bfseries{4} & \bfseries{ (0,1), (4,0) }\\ \hline
5 & (0,2), (1,1), (5,0)\\ \hline
6 & (0,2), (1,1), (6,0)\\ \hline
7 & (0,2), (1,1), (7,0)\\ \hline
8 & (0,2), (2,1), (8,0)\\ \hline
9 & (0,2), (3,1), (9,0)\\ \hline
\bfseries{10} & \bfseries{ (0,2), (4,1), (10,0) }\\ \hline
11 & (0,3), (1,2), (5,1), (11,0)\\ \hline
\bfseries{12} & \bfseries{ (0,3), (1,2), (6,1), (12,0) }\\ \hline
13 & (0,3), (1,2), (7,1), (13,0) \\ \hline
   \end{tabular}
    \captionlistentry[table]{A sample Graph $G$ and its adjacency list}
    \captionsetup{labelformat=andtable}
    \caption{A sample Graph $G$ and the created hub-labels}
    \label{fig:graph}
  \end{figure}

The offline phase of the \ReHub\ algorithm takes place after the creation of the hub labels. Although the \ReHub\ algorithm works with any correct Hub Labeling algorithm, in this work we generate the necessary labels using the PLL algorithm of~\cite{akiba2013f}, as provided by its authors in~\cite{akiba2013fcode}. 
To highlight the results of the PLL algorithm, the generated labels for the sample undirected, unweighted graph $G$ of Figure~\ref{fig:graph} are shown in Table~\ref{fig:graph}. 
In the remainder of this work we will refer to those labels as the \emph{forward labels}. 
We also assume that the target objects are located at vertices 4,10,12, i.e., $P=\{4,10,12\}$. The respective entries are highlighted in Table~\ref{fig:graph}. 
For each vertex~$v$, the forward label $L(v)$ is an array of pairs $(u,dist(v,u)$ sorted by hub vertex $u$. This is the starting point for the offline phase of the \ReHub\ algorithm, which in turn is divided in three smaller substages: (i)~the \emph{kNN backward labels} construction, (ii)~the \emph{batch \kNN\ calculations} from all objects, and (iii)~the \RkNN\ \emph{backward labels} construction. 
Each of these stages will be described in the following.  

\subsubsection{The \kNN\ backward labels construction}
\label{subsub:knnTs} 

To efficiently answer one-to-many queries with hub labels, we need to store separately the hub labels of the target objects $P=~\{P_1, \dots P_i,\dots P_{|p|}\}$ ordered by hubs~\cite{delling2011g}. For each such hub $u$, those \emph{backward labels-to-many} is an array of pairs $(P_i,d(u,P_i)$. 
Expanding this approach for \kNN\ queries, \cite{abraham2012hldb}~showed that if we know the number~$k$ in advance (or the maximum $k$ we will service for \kNN\ queries) , then for each hub we only need to keep the $k$-pairs with the smallest distances per hub. 
Although these previous works focused on road networks, the correctness of this approach still applies to undirected, unweighted graphs. This process for the sample graph and $k=2$ is shown in Table~\ref{tab:knn-bck}. For small-diameter graphs (like the ones used in this work) we will have many ties (in terms of distance), but keeping at most $k$-labels still ensures correctness. 

\begin{table}[tb]
\centering
\caption{The \ReHub\ \kNN\ backward labels creation for the sample graph $G$, $k=1$ and \\$P=\{4,10,12\}$}
\small{
\begin{tabular}{|c|c|c|c|}  \hline 
\bfseries{} & \bfseries{Backward Labels} & \bfseries{\kNN\ Backward Labels} & \bfseries{\ReHub\ \kNN\ } \\ 
\bfseries{Hub} & \bfseries{ (to-many) } & \bfseries{ (k=2) } & \bfseries{Backward Labels  (k=1) } \\ \hline
0 & (4,1), (10,2), (12,3) & (4,1), (10,2)  & (0,1), (1,2) \\ \hline
1 & (12,2) & (12,2) &  (2,2 )\\ \hline
4 & (4,0), (10,1) & (4,0),(10,1) & (0,0), (1,1) \\ \hline
6 & (12,1)& (12,1)  & (2,1) \\ \hline
10 &(10,0) & (10,0) & (1,0)\\ \hline
12 & (12,0) & (12,0) & (2,0) \\ \hline
\end{tabular}   
}
 \label{tab:knn-bck}
\end{table}


\begin{codebox}
\Procname{$\proc{{knnLab}}\>(\id{P}, |P|, \id{k}, forwLabels,  kNNLab)$}
\li Initialize$(kNNLab,(|V|, BoundPQue(k+1)) )$
\li \For $i \gets 0 $ \To $|P|$
\li \Do
 \For $j \gets 0$ \To $forwLabels[P[i]].size$
\li \Do
$hub \gets forwLabels[P[i]][j].hub$
\li $d \gets forwLabels[P[i]][j].dist$
\li $kNNLab[hub].push(i,d)$
\End
\End
\End
\end{codebox}

Due to the pruning of the PLL algorithm, in our example, \emph{\kNN\ backward labels} do not necessarily have as many as $k$-pairs per hub. To create the \kNN\ backward labels for \rehub, we need to do some additional modifications. (i)~When answering \RkNN\ queries, we must assume that $k=k+1$ during the construction of the \kNN\ backward labels. This is necessary, since in our example the NN of object $10$ (for $k=1$) is by definition the same object, but for \RkNN\ queries with $k=1$, the NN neighbor of $10$ is object $4$. (ii)~Instead of storing the vertex IDs $P_i$ of the objects in the \kNN\ backward Labels, we store the array index $i$ of each object, as shown in the last column of Table~\ref{tab:knn-bck}. This facilitates faster processing during the remaining substages of the offline and online phase of the \ReHub\ algorithm. 
On the technical side, the \kNN\ backward labels creation is quite fast, since we only have to loop through the forward labels of the objects in $P$ and use a vector-based bounded priority queue of size $k+1$ per hub to calculate the $k+1$ pairs with the smallest distances per hub. 
This method offers two major advantages. (i)~We do not need to build the intermediate backward labels-to-many data structure (column~2, Table~\ref{tab:knn-bck}), which would be much slower, and (ii)~when looping through the forward labels of each object, pairs with distances greater than the $k+1$ worst distance previously found for a specific hub may be safely ignored.

The pseudocode for the \kNN\ backward labels construction is shown in procedure \proc{knnLab} and throughout this process, for each hub  we use a vector-based bounded priority queue of size $k+1$ that stores pairs in the form $(idx,dist)$ ordered by distance.

\subsubsection{Batch \kNN\ calculations from objects}
\label{subsub:knnBatch} 

After creating the \kNN\ backward labels (column 4, Table~\ref{tab:knn-bck}), we need to calculate the $k$-nearest neighbors of each object. To this end, we perform a total of $|P| \times kNN$ calculations, using the created \kNN\ backward labels. 
Each of those \kNN\ computations uses the method implicitly described in~\cite{abraham2012hldb}, with the additional constraint that for each object when traversing the \kNN\ backward labels of one of its hubs, we skip the labels corresponding to this specific object index. 

\begin{codebox}
\Procname{$\proc{{BatchKnnCalc}}\>(\id{P}, |P|, \id{k}, forwLabels, kNNLab, kNNResults)$}
\li Initialize$(kNNResults,(|P|, BoundPQue(k)) )$
\li  \Parfor $i \gets 0 $ \To $|P|$
\li \Do
 \For $j \gets 0$ \To $forwLabels[P[i]].size$
\li \Do
$hub \gets forwLabels[P[i]][j].hub$
\li $d \gets forwLabels[P[i]][j].dist$
\li \For $k \gets 0$ \To $kNNLab[hub].size$
\li \Do
$idx \gets kNNLab[hub][k].idx$
\li \If $idx != i  $
\li \Then
 $d2 \gets d+kNNLab[hub][k].dist$
\li $kNNResults[i].push(idx,d2)$
\End	
\End
\End
\End
\End
\end{codebox}

The simplified pseudocode for the batch \kNN\ calculations from objects is shown in procedure~\proc{BatchKnnCalc}. The $kNNResults$ are also stored in a $|P|$-sized vector of  vector-based bounded priority queues of size $k$ that store pairs in the form $(idx,dist)$ ordered by distance.  For each such object, when traversing the \kNN\ backward labels of one of its hubs, we skip the pairs corresponding to the index of this specific object (Line~8 in the pseudocode). Moreover, every time a new pair is pushed to the corresponding queue (Line~10), our customized push operation checks if the ``pushed'' object index already exists in the queue with a smaller or equal distance value than the pushed pair. If yes, we can safely ignore this pair. If, on the other hand, this object index exists in the queue with a larger distance value, we update this distance value and resort the queue. If the pushed object index does not already exist in the queue, our custom push operation checks if the queue has less than $k$ items. In that case, the new pair enters the queue and the queue is resorted. If the queue has already $k$ items, our push operation checks if the new pair is better (i.e.,~corresponds to a smaller distance) than the last ($k$) element of the queue. If yes, the last element is popped, the new pair enters the queue at the end and the queue is resorted. Since each queue is basically a vector of size $k$, popping back, pushing back and resorting this (rather small) priority queue are very fast operations. 

We can further accelerate the process, if every time a new pair $(idx,d2)$ enters the $kNNResults[i]$ queue for a specific object, we check if the queue already has $k$-items; In that case we store the worst label distance as a separate variable. If the distance $d$ (Line~5) or the distance $d2$ (Line~9) are greater than this worst distance, we can safely skip this particular pair. Especially, in the second case (distance $d2$ - Line 9) we can exit the third loop (Line~6) completely, since the \kNN\ backward label of each hub is ordered by distance. This optimization (not shown in the pseudocode for readability) accelerates significantly each individual \kNN\ calculation. 

\begin{table}[tb]
\centering
\caption{Batch kNN calculations process for the sample graph $G$, $k=1$ and \\$P=\{4,10,12\}$}
\small{
\begin{tabular}{|c|c|c|c|c|}  \hline 

\bfseries{ Obj. } & \bfseries{Forward Labels} & & \bfseries{ReHub kNN }  & \bfseries{ kNN Results} \\ 
\bfseries{ ID} & \bfseries{of Objects} & \bfseries{  Hub } &\bfseries{Backward Labels (k=1)} & \bfseries{ (idx, dist) } \\ \hline
\multirow{2}{*} {\bfseries{ 4 }} & \multirow{2}{*} { (0,1), (4,0) \tikzmark{a} } & \tikzmark{1} { 0 }  & (0,1), (1,2) & \multirow{2}{*} {\bfseries{ (1,1) }}\\ \cline{3-4}
  & & \tikzmark{2} { 1 } &  (2,2 ) &\\ \hline
\multirow{2}{*} {\bfseries{ 10 }} & \multirow{2}{*} { (0,2), (4,1), (10,0) \tikzmark{b} } & \tikzmark{3} { 4 } & (0,0), (1,1) & \multirow{2}{*} {\bfseries{ (0,1) }} \\ \cline{3-4}
&&  \tikzmark{4} { 6}  &   (2,1) & \\ \hline
\multirow{2}{*} { \bfseries{12}} & \multirow{2}{*} { (0,3), (1,2), (6,1), (12,0) \tikzmark{c} }& \tikzmark{5} { 10 } & (1,0) & \multirow{2}{*} { \bfseries{ (0,4) } }\\  \cline{3-4}
& & \tikzmark{6} { 12 } & (2,0) &\\ \hline
\end{tabular}

\linkkk{a}{1}
\linkkk{a}{3} 
\linkk{b}{1}
\linkk{b}{3} 
\linkk{b}{5} 
\link{c}{1} 
\link{c}{2}
\link{c}{4} 
\link{c}{6}
}
 \label{tab:knn-batch}
\end{table}

The results of this process are shown on Table \ref{tab:knn-batch}, where the combination of the forward labels of the objects $\{4,10,12\}$ with the  \kNN\ backward labels shows that the \kNN\ of object~$4$ is the object with index~$1$, i.e., object~$10$, with distance~$1$.
The \kNN\ of object $10$ is the object with index~$0$ (object~$4$) with the respective distance~$1$ and finally, the \kNN\ of object $12$ is the object with index~$0$ (object~$4$) with the respective distance~$4$. To facilitate faster computation, each \kNN\ computation may be performed in parallel (Line~2 of procedure~\proc{BatchKnnCalc}) since there is no interaction between the individual \kNN\ calculations. Considering this is the slower substage of the offline phase, employing parallelism significantly drops the total preprocessing time required for the \emph{ReHub}'s offline phase. 

\subsubsection{The RkNN backward labels construction}
\label{subsub:rknnTs} 

After calculating the \kNN\ of each object, for answering \RkNN\ queries it would suffice to run an \emph{one-to-many} HL query from the query vertex $q$ to all objects, by constructing and using the \emph{backward labels-to-many} of objects $P$ (see column 2, Table~\ref{tab:knn-bck}) and then loop through the calculated distances to see if they are smaller or equal to the \kNN\ distances calculated by the previous step. 
But we can do much better: We construct an alternative data structure, referred hereafter as the \emph{RkNN backward labels}, based on the observation that \emph{we need to calculate distances to a specific object, if and only if those distances are equal or smaller than the distance of the \kNN\ of this object}. If the objects are uniformly distributed through the graph, this optimization ensures that only hubs of relatively small distances from each object are added to the \RkNN\ backward labels. Therefore, during the online phase, if the query vertex $q$ is faraway from some objects, there would be no matching hubs between those objects and the query vertex. 

\begin{table}[tb]
\centering
\caption{\RkNN\ backward labels construction for the sample graph  $G$, $k=1$ and \\$P=\{4,10,12\}$}
\small{
\begin{tabular}{|c|c|c|c|c|}  \hline 
\bfseries{ Obj. } & \bfseries{ kNN Result } & \bfseries{ Forward Labels }  & & \bfseries{ RkNN Backward Labels }   \\ 
\bfseries{ ID} & \bfseries{ (idx, dist) } & \bfseries{ of Objects }   & \bfseries{  Hub } &\bfseries{ (k=1) } \\ \hline
\multirow{2}{*} {\bfseries{ 4 }} & \multirow{2}{*} { (1,\textbf{1}) } & \multirow{2}{*} { (0,1), (4,0) }   &  { 0 }  &  (0,1), (2,3)  \\ \cline{4-5}
  & & & { 1 } &  (2,2 ) \\ \hline
\multirow{2}{*} {\bfseries{10} } & \multirow{2}{*} { (0,\textbf{1}) } & \multirow{2}{*} { \sout{(0,2)}, (4,1), (10,0) }   &  { 4 } & (0,0), (1,1) \\ \cline{4-5}
&&&  { 6 }  &   (2,1)  \\ \hline
\multirow{2}{*} { \bfseries{12} }  & \multirow{2}{*} { (0,\textbf{4}) } &\multirow{2}{*} { (0,3), (1,2), (6,1), (12,0) }     & { 10 } & (1,0)\\  \cline{4-5}
&& &{ 12 } & (2,0) \\ \hline
\end{tabular}
}
 \label{tab:rknn-bck}
\end{table}

The resulting pseudocode for the \RkNN\ backward labels construction is shown in procedure~\proc{RknnLab} and the entire process is highlighted in Table~\ref{tab:rknn-bck}. When we build the \RkNN\ labels for object $10$, we skip the pair $(0,2)$ because the \emph{NN} of object $10$ is within distance of $1$ and therefore pairs with greater distances than that (for this particular object) may be safely ignored. Again, when building the \RkNN\ backward labels we use the objects array indexes, instead of their IDs. 

\begin{codebox}
\Procname{$\proc{{RknnLab}}\>(\id{P}, |P|, \id{k}, forwLabels, kNNResults,RkNNLab)$}
\li Initialize$(RkNNLab,(|V|, vector<$(idx,dist)$>) )$
\li \For $i \gets 0 $ \To $|P|$
\li \Do
 \For $j \gets 0$ \To $forwLabels[P[i]].size$
\li \Do
$d \gets forwLabels[P[i]][j].dist$
\li \If $d <= kNNResults[i][k-1]  $
\li \Then
$hub \gets forwLabels[P[i]][j].hub$
\li $RkNNLab[hub].push\_back(i,dist)$
\End
\End
\End
\End
\end{codebox}

Several interesting observations can be made by comparing Tables~\ref{tab:knn-bck} and~\ref{tab:rknn-bck}. Firstly, as expected, the number of \RkNN\ backward labels (column~5, Table~\ref{tab:rknn-bck})  is smaller than the backward labels-to-many (column~2, Table~\ref{tab:knn-bck}). Although for our small sample graph~$G$ this difference is minimal, for larger graphs it becomes significant.
Therefore, using the \RkNN\ backward labels will significantly improve the online phase of the \ReHub\ algorithm. This will be clearly showcased in our experimentation presented in Section~\ref{sec:experiments}.
Second, the \kNN\ backward labels (column 4, Table~\ref{tab:knn-bck}) are different than \RkNN\ backward labels (column 5, Table~\ref{tab:rknn-bck}). The added benefit is that by using the \kNN\ backward labels we can still answer \kNN\ queries and by using the \RkNN\ backward labels we can answer \RkNN\ queries within the same framework.     

\subsection{Online Phase}
\label{sub:online} 

The offline phase of the \ReHub\ algorithm runs only once for a specific set of objects~$P$. Its final output is (i)~a matrix of size $|P| \times k$ of (ordered by distance per row) pairs $(idx, dist)$ that contain the \kNN\ of each object and (ii)~the \RkNN\ backward labels. 
The following online phase of the \ReHub\ algorithm is basically a modified one-to-many HL query from the query vertex $q$ that operates on the \RkNN\ backward labels and is described by the pseudocode of procedure~\proc{OnlinePhase}. 
The output of the online phase is a vector (denoted $out$ in the pseudocode) of size $|P|$ with all values set to infinity, except those that belong to the indexes of the objects of the \RkNN\ set; those values are set to the correct distances from query vertex $q$ to the respective objects. 
In our running example of the sample graph $G$, $P=\{4,10,12\}$ and $k=1$, the online phase for a \RkNN\ query from vertex $0$ would only have to visit the \RkNN\ backward labels of hub $0$ (see Tables~{\ref{fig:graph} and~\ref{tab:rknn-bck}) and would output the result $out =\{1,\infty,3\}$, meaning that the objects $4$, $12$ belong to the \RkNN\ set of vertex $0$ with distances $1$ and $3$ respectively.

\begin{theorem}
\label{theorem1}
The \ReHub\ algorithm is correct.
\end{theorem}

\begin{proof}
Building the \kNN\ backward labels and then performing the batch \kNN\ calculations to calculate the \kNN\ of each object is correct, because it follows the methodology of Abraham et al.~\cite{abraham2012hldb} who proved its correctness. 
Building the \RkNN\ backward labels is also correct, since we just reorder all labels of the objects according to hub, except those that correspond to distances greater than the \kNN\ of its object. This ensures than we can calculate correct distances to any of those objects from any query vertex, except when this query vertex is farther than the \kNN\ of a specific object. The online phase is also correct, since it operates on the \RkNN\ backward labels and updates the result vector $out$ for a specified object, only when the calculated distance is smaller or equal than the distance of the \kNN\ of this object (Line 8,  procedure~\proc{OnlinePhase}). Therefore the \ReHub\ algorithm is also correct.   
\end{proof}

\begin{codebox}
\Procname{$\proc{{OnlinePhase}}\>(q,\id{P}, |P|, \id{k}, forwLabels,  $  $kNNResults, RkNNLab, out)$}
 \li Initialize$(out,(|P|, \infty))$
\li \For $i \gets 0 $ \To $forwLabels[q].size$
\li \Do
$hub \gets forwLabels[q][i].hub$
\li $d \gets forwLabels[q][i].dist$
\li \For $j \gets 0$ \To $RkNNLab[hub].size$
\li \Do
	$idx \gets RkNNLab[hub][j].idx$
	\li $d2 \gets d + RkNNLab[hub][j].dist  $
	\li \If $d2 < out[idx]  \And $ \\ 
	\> \>  $d2 \leq kNNResults[idx][k-1].dist $
\li \Then
$out[idx] \gets d2$ 
\End	
\End
\End
\End
\end{codebox}

The main advantage of the \ReHub\ algorithm, in comparison to previous works, is the separation between the costlier offline phase, which runs only once for a specific set of objects and the very fast online phase. An additional benefit of \ReHub\ compared to the works of~\cite{yiu2006,borutta2014} is that not only \ReHub\  calculates the \RkNN\ set of the query vertex but it also \emph{calculates the correct network distances from the query vertex to any of the objects belonging in the \RkNN\ set}. Regarding the online phase, operating on the \RkNN\ backward labels is significantly faster, since for large graphs those \RkNN\ backward labels are significantly fewer than the backward labels-to-many. Also the usage of object array indexes instead of the object IDs accelerates the whole process, since the final results vector $out$ is of size $|P|$ instead of $|V|$ which makes its initialization faster (Line~1, procedure~\proc{OnlinePhase}). Also, accessing the \kNN\ results of each object (Line~8) and the previous best value of results table (Line~8 and~9) are very cheap operations, since they operate on smaller vectors of size $|P|$. Moreover, the memory required for storing these intermediate data structures is also significantly smaller. This will be further quantified in the next section, where we analyze the complexity and memory requirements of the \emph{ReHub} algorithm.

\subsection{Complexity Analysis and Memory Requirements}
\label{sub:complexity_analysis}

If $D$ is the object density defined as $D=\frac{|P|}{|V|}$, then the number of objects is $D\cdot|V|$.  The forward label of each vertex has an average of $\frac{|HL|}{|V|}$ hubs, where $|HL|$ is the total number of labels created by the hub-labeling algorithm (PLL in our case). Especially in the case of the PLL algorithm this approximation is pretty accurate, since Akiba et al.~\cite{akiba2013f} have shown that the "\emph{size of the created labels does not differ much for different vertices and few vertices have much larger labels than the average}". Since we have  $D \cdot |V|$ objects and  $\frac{|HL|}{|V|}$ hubs per object, then the backward labels-to-many will have on average $D \cdot |HL|$ pairs. Regarding the offline phase, the \kNN\ backward labels construction needs to access  all those $D \cdot |HL|$ pairs (same as the backward labels-to-many) to construct the \kNN\ backward labels that have a maximum of $k+1$ pairs per hub. 
In the batch \kNN\ calculations, we have a total of  $D \cdot |V|$ $k$NN queries that each needs to access on average $ (k+1) \cdot \frac{|HL|}{|V|}$ pairs to create the \kNN\  results of size of $k \cdot D \cdot |V|$. Therefore, the complexity of the batch \kNN\ calculations will be $(k+1) \cdot D \cdot |HL|$. Finally, for the \RkNN\ backward labels construction we need to access $D\cdot |HL|$ pairs (same as the backward labels-to-many) and the  $k \cdot D \cdot |V|$ results (to retrieve the worst $k$ label per object). Conclusively, both the \kNN\ and \RkNN\ backward labels construction have a complexity of  $D \cdot |HL|$ each (since $|HL| >> |V|$), where the most costly batch \kNN\ calculations stage has a complexity of $(k+1) \cdot D \cdot |HL|$.   

Regarding the online phase, for a very large number of $k$, the online phase of \emph{ReHub} will degrade to a one-to-many query between the query vertex $q$ and the set of objects~$P$. Therefore, we will first analyze the complexity of an one-to-many HL query. As showed earlier,  the backward labels-to-many will have on average $D \cdot |HL|$ pairs. If those pairs are equally distributed per hub, then each hub on the  backward labels-to-many will have an average of  $D \cdot \frac{|HL|}{|V|}$ pairs. Since the forward label of the  query vertex~$q$ will have on average of $\frac{|HL|}{|V|}$ hubs, an one-to-many query from the query vertex will access on average $D \cdot ({\frac{|HL|}{|V|}})^2$ pairs. Thus, the online phase of \emph{ReHub} will access  $\varepsilon \cdot D \cdot ({\frac{|HL|}{|V|}})^2$ pairs, where $\varepsilon < 1$ (since the size of the \RkNN\ backward labels is smaller than the backward labels-to-many) and $\varepsilon = f(k,D)$, i.e., the value of $\varepsilon$ depends on the density $D$ and the cardinality $k$ of the \RkNN\ results. In fact, our experimentation have showed that $\varepsilon$ becomes smaller for larger values of $D$ and smaller values of $k$. The aforementioned theoretical results are summarized in Table \ref{tab:complexity} where we also report the memory required for storing the results of each stage, considering the fact that each pair requires $5$ bytes for storage ($4$ bytes for object index + $1$ byte for distance) and the result of the online phase is a sized $D \cdot |V|$ vector of distances.

\begin{table}[tb]
\centering
\caption{\emph{ReHub} complexity and memory requirements}
\begin{tabular}{|c|c|c|}  \hline 
\bfseries{  } & \bfseries{  } & \bfseries{ Memory for  } \\  
\bfseries{ Stage } & \bfseries{ Complexity } & \bfseries{ storing result (B) } \\  \hline 
\kNN\ backward labels construction &  $D \cdot |HL|$  & $ 5 \cdot (k+1) \cdot |V| $  \\ \hline 
Batch \kNN\ calculations &  $(k+1) \cdot D \cdot |HL|$  &  $5 \cdot k \cdot D \cdot |V|$  \\ \hline 
 \RkNN\ backward labels construction & $D \cdot |HL|$  &  $ 5 \cdot \varepsilon \cdot  D \cdot |HL| $  \\ \hline 
 Online Phase &  $\varepsilon \cdot D \cdot ({\frac{|HL|}{|V|}})^2$   & $D \cdot |V|  $  \\ \hline 
\end{tabular}
 \label{tab:complexity}
\end{table}      

Our theoretical evaluation shows that even for large values of $k$ where the online phase of \emph{ReHub} would converge to an one-to-many query, \emph{ReHub}'s online performance will remain excellent, as long as the fraction $\frac{|HL|}{|V|}$ is relatively small, i.e., the number of created labels is proportionate to the number of graph vertices. As our experimentation showed (see Section~\ref{sec:experiments}), this fraction $\frac{|HL|}{|V|}$  was below 5,000 for all network graphs we experimented with.

\subsection{Extension to Directed and Weighted Graphs}
\label{sub:directed_graphs} 

Throughout this work and the experimentation described in Section~\ref{sec:experiments}, we use undirected and unweighted graphs. However, the \emph{ReHub} algorithm may be easily extended to directed graphs with the following changes: (i) In the offline phase the $k$NN backward labels must be constructed from the backward labels (ii) In the online phase we must use the backward labels of query vertex $q$. Note that \emph{most previous methods like \cite{yiu2006,borutta2014} have only been applied on undirected networks}. For weighted graphs, \emph{ReHub} will work without requiring any further modifications.

\section{Experiments}
\label{sec:experiments} 


To evaluate the performance of \ReHub\ on various large-scale graphs, we conducted experiments on a workstation with a 4-core Intel i7-4771 processor clocked at 3.5GHz and 32 GB of RAM, running Ubuntu 14.04. Our code was written in C++ with GCC~4.8 and optimization level 3. We used OpenMP for parallelization. 

The network graphs used in our experiments were taken from the Stanford Large Network Dataset Collection~\cite{snapnets} and the 10th Dimacs Implementation Challenge website~\cite{dimacs2012}. All graphs are undirected, unweighted and strongly connected. We used colla-boration  graphs (DBLP, Citeseer1, Citeseer2)~\cite{geisbergerss2008}, social networks (Facebook~\cite{mcauley2012}, Slashdot1 and Slashdot2~\cite{leskovec2009}), networks with ground-truth communities (Amazon, Youtube)~\cite{yang2012b}, web graphs (Notre Dame)~\cite{reka1999} and location-based social networks (Gowalla)~\cite{cho2011}. 
The graphs' average degree is between $3$ and $37$ and the PLL algorithm creates $26 - 4,457$ labels per vertex, requiring $0.03 - 5,950s$ for the hub labels' construction (see~Table~\ref{tab:graph_stats}). For each individual \RkNN\ experiment we generate randomly $100$ sets of objects of size $|P|$ and then we generate $100$ random query points per set. \RkNN\ query times are then averaged over those 10,000 experiments.

\begin{table}[tb]
\centering
\caption{Networks graphs statistics}
{\small
\begin{tabular}{|c|c|c|c|c|c|}  \hline 
\bfseries{Graph} & \bfseries{ \textbar~V \textbar } & \bfseries{ \textbar~E \textbar} & \bfseries{ Avg degr. } & \bfseries{ \textbar~HL \textbar~/ \textbar~V \textbar } & \bfseries{PLL Preproc. Time (s)}\\ \hline
\bfseries{ Facebook } & 4,039 & 88,234 & 22 & 26 & 0.03\\ \hline
\bfseries{ NotreDame } & 325,729 & 1,090,108 & 3 & 55 & 6\\ \hline
\bfseries{ Gowalla  } & 196,591 & 950,327 & 5 & 100 & 13\\ \hline
\bfseries{ Youtube } & 1,134,890 & 2,987,624 & 3 & 167 & 123\\ \hline
\bfseries{ Slashdot0811  } & 77,360 & 469,180 & 6 & 204 & 11\\ \hline
\bfseries{ Slashdot0922  } & 82,168 & 504,230 & 6 & 216 & 13\\ \hline
\bfseries{ Citeseer1 } & 268,495 & 1,156,647 & 4 & 408& 110\\ \hline
\bfseries{ Amazon } & 334,863 & 925,872 & 3 & 689 & 230\\ \hline
\bfseries{ DBLP } & 540,486 & 15,245,729 & 28 & 3,628 & 5,720\\ \hline
\bfseries{ Citeseer2 } & 434,102 & 16,036,720 & 37 & 4,457 & 5,946\\ \hline
\end{tabular}  
}
 \label{tab:graph_stats}
\end{table}

\subsection{Performance and Memory Requirements}
\label{sub:performance} 

In our first round of experiments we evaluate the performance of \ReHub\ in comparison to the object density $D=|P| / |V|$. Figure~\ref{fig:d_k_1} reports the time required for the offline and online phases of \ReHub\ for $k=1$ and $D = \{0.001, 0.005, 0.01, 0.05, 0.1\}$, similar to the methodology followed in~\cite{yiu2006}. For the offline phase, we parallelized only the slower substage of batch \kNN\ computations from objects. Online phase is always sequential. 

\begin{figure}[!tb]
\centering
 \subfigure[Offline phase. Y-axis is on logarithmic scale and time is reported in $ms$] 
{\includegraphics[width=0.497\columnwidth]{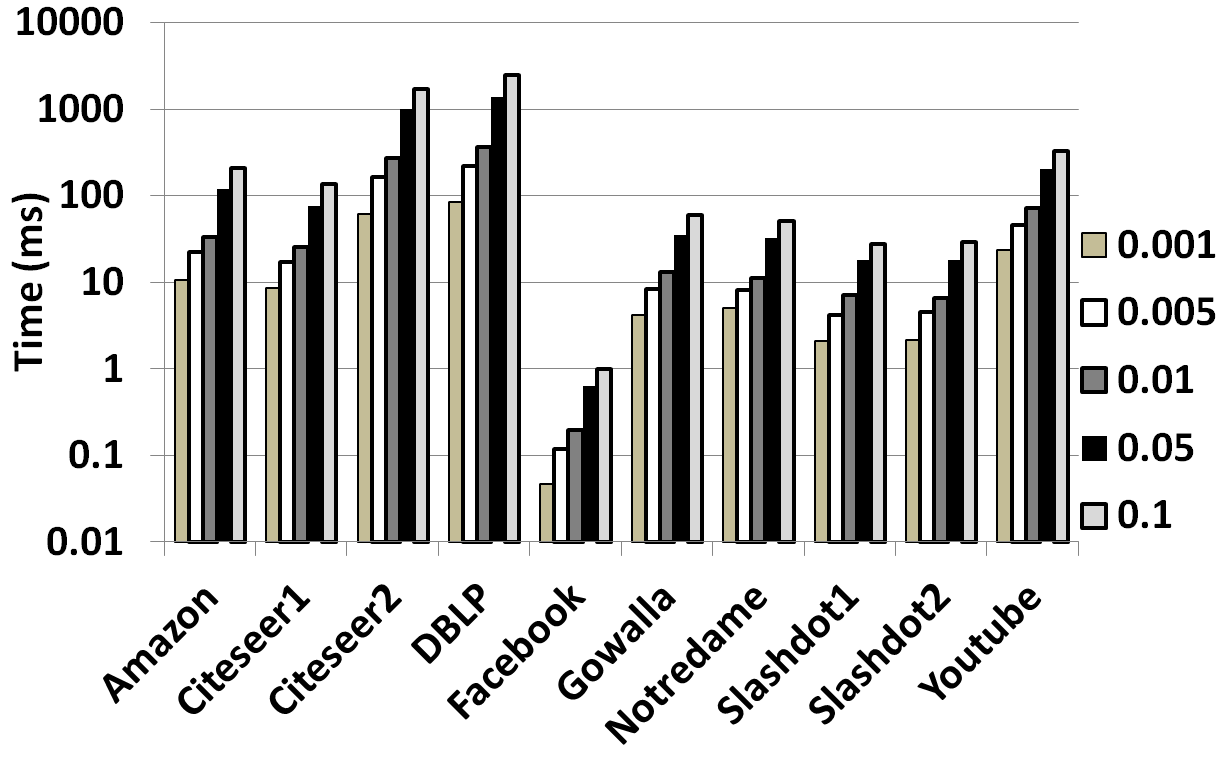}\label{fig:off1}} 
 \subfigure[Online phase. Y-axis is on logarithmic scale and time is reported in $\mu s$]
{\includegraphics[width=0.497\columnwidth]{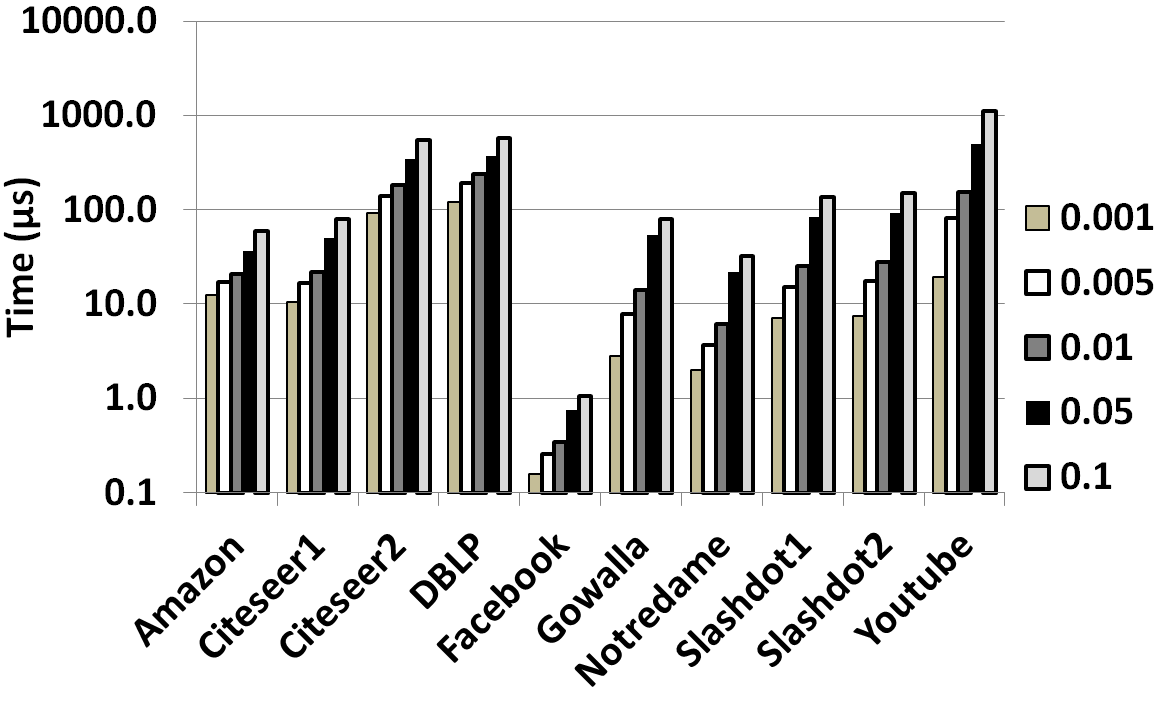}\label{fig:on1}}
 \caption{Offline and online phases of \ReHub\ for $k=1$ and varying values of $D$}
 \label{fig:d_k_1}
\end{figure}  

We see that the offlne phase of \ReHub\ for most graphs takes less than $250ms$ for all values of $D$. Even for the worst performing graphs (DBLP, Citeseer2) it takes less than $1s$, except for very dense distributions of objects ($D=0.1$). Considering the fact, that such dense distribution of objects is not common and comparing the offline phase's time with the construction of the hub labels (Table~\ref{tab:graph_stats}), even the time required for the offline phase could be considered negligible. Regarding the online phase, results are even more impressive. For all graphs and all values of $D$, the online phase takes less than $0.6ms$, except for the Youtube graph and $D=0.1$ ($1.1ms$). 
Generally speaking, both phases perform worse for increasing values of $D$ (and $|P|$) and for larger number of labels per vertex (e.g. DBLP, Citeseer2), as predicted by the theoretical analysis of the $ReHub$ algorithm presented in Section~\ref{sub:complexity_analysis}.

In terms of memory requirements, Figure~\ref{fig:memory1} reports the memory required for storing the additional data structures for \ReHub\ (\kNN\ backward labels, \kNN\ results per object and the \RkNN\ backward labels) and the size of the \RkNN\ backward labels in comparison to the backward labels-to-many, for the same setting as our previous experiment (i.e., for $k=1$ and $D = \{0.001, 0.005, 0.01, 0.05, 0.1\}$). Results show that the memory required for the additional data structures for \ReHub\ is less than $13Mb$ even for the worst performing graphs (DBLP, Citeseer2) and the size of the \RkNN\ backward labels can be more than $100 \times$ smaller than the size of the backward labels-to-many for very dense objects ($D=0.1$), i.e.,  the corresponding online phase there will be consequently $100 \times$ faster than an one-to-many query. This shows that even for such large values of the density $D$ (which constitutes the worst-case scenario for $ReHub$) its online phase will still remain very fast and efficient. 

\begin{figure}[!tb]
\centering
 \subfigure[Index size (Mb) for $ReHub$. Y-axis is on logarithmic scale] 
{\includegraphics[width=0.497\columnwidth]{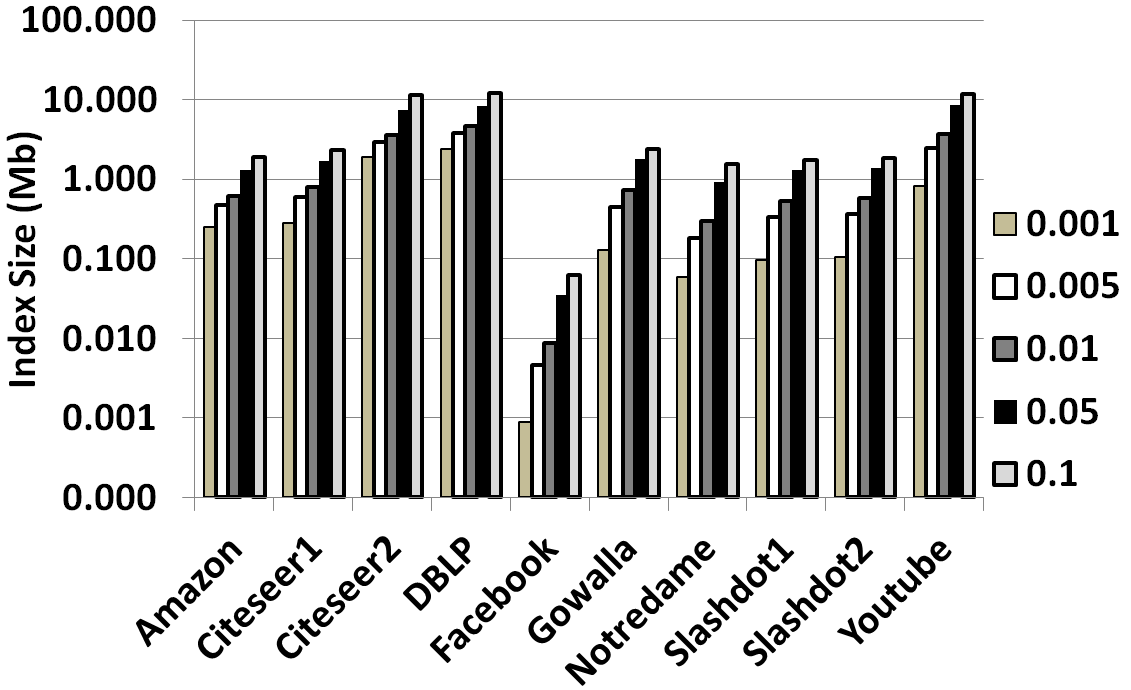}\label{fig:memory1a}} 
 \subfigure[Size of \RkNN\ backward labels in comparison to backward labels-to-many. Y-axis is on logarithmic scale]
{\includegraphics[width=0.497\columnwidth]{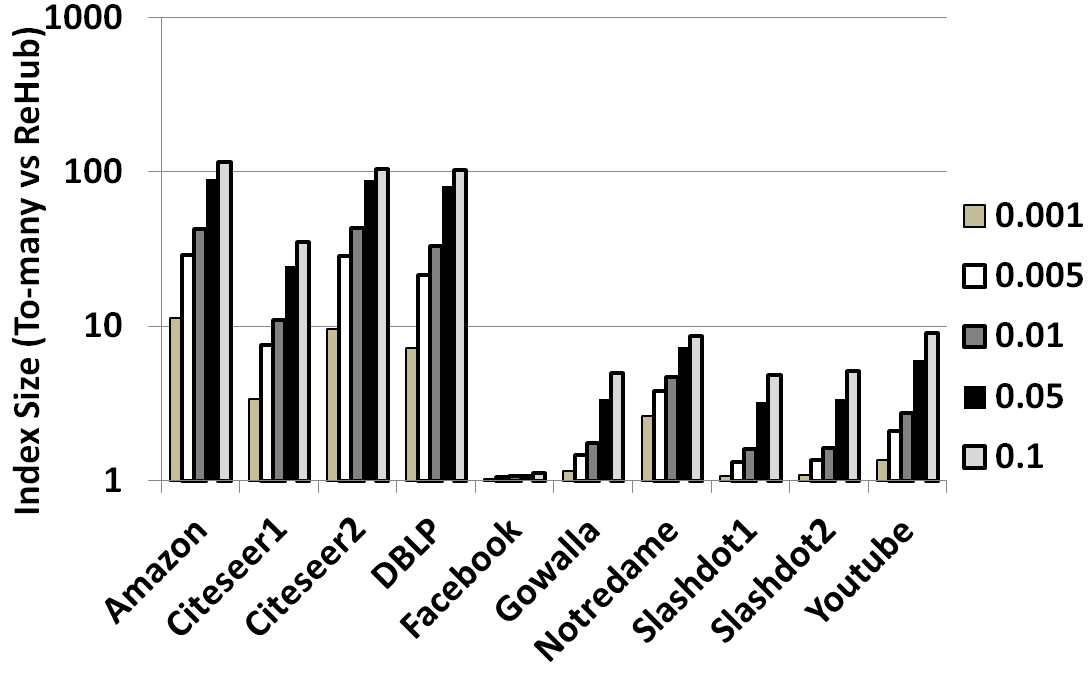}\label{fig:memory1b}}
 \caption{Memory footprint of \ReHub\ for $k=1$ and varying values of $D$}
 \label{fig:memory1}
\end{figure}  

In our second round of experiments, similar to~\cite{yiu2006}, we assess the performance and memory characteristics of the \ReHub\ algorithm in comparison to $k$. To this end, Figure~\ref{fig:k_d_0_01} reports the time required for the offline and online phases of the \ReHub\ algorithm for $D=0.01$ and $k = \{1, 2, 4, 8, 16,32\}$. Again, for the offline phase, we parallelized only the batch \kNN\ computations from objects.  

\begin{figure}[!tb]
\centering
 \subfigure[Offline phase. Y-axis is on logarithmic scale and time is reported in $ms$]
{\includegraphics[width=0.497\columnwidth]{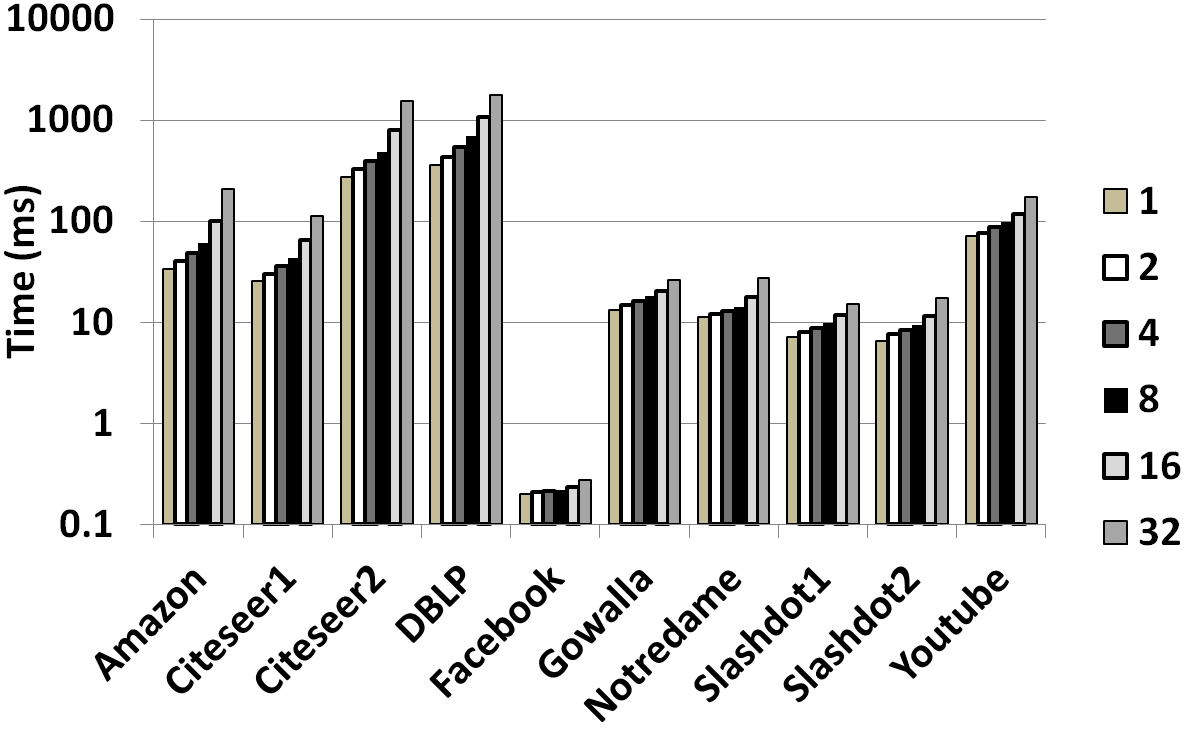}\label{fig:off2}} 
 \subfigure[Online phase. Y-axis is on logarithmic scale and time is reported in $\mu s$]
{\includegraphics[width=0.497\columnwidth]{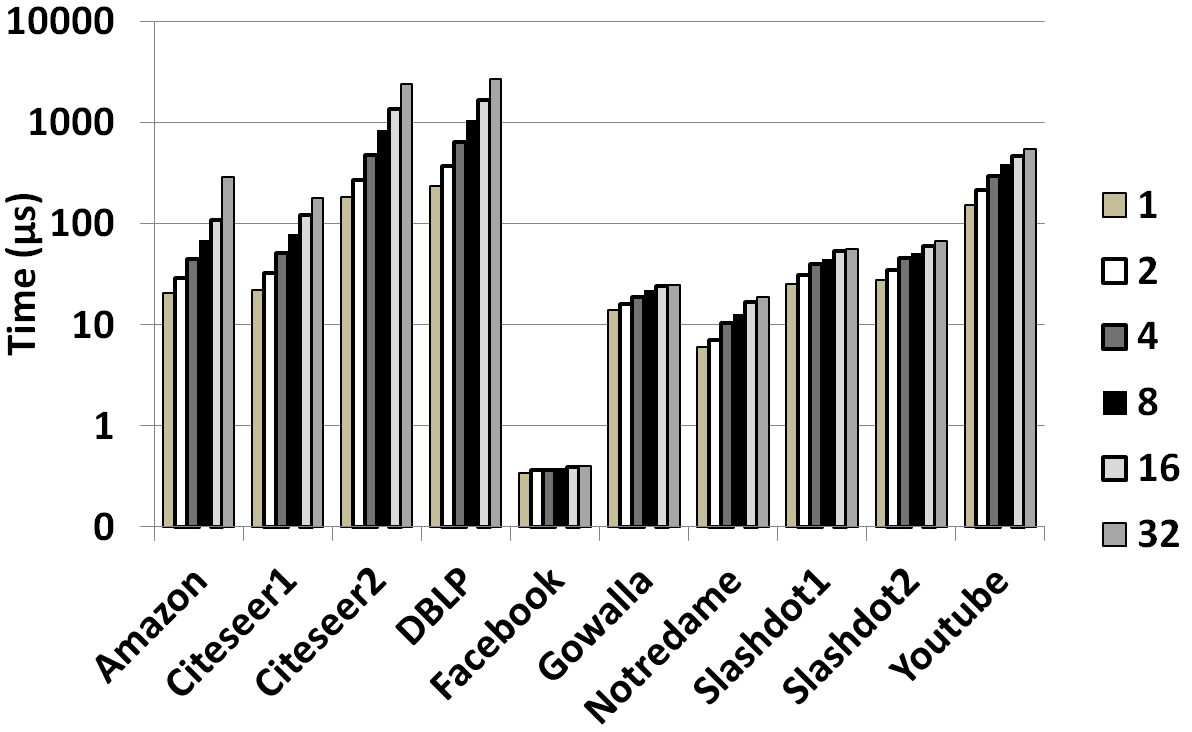}\label{fig:on2}}

 \caption{Offline and online phases of \ReHub\ for $D=0.01$ and varying values of $k$}
 \label{fig:k_d_0_01}
\end{figure}  

As before, the offlne phase of \ReHub\ takes less than $250ms$ for most graphs and values of~$k$. Even for the worst performing graphs (DBLP, Citeseer2) it takes less than~$1.1s$, except for $k=32$ ($1.7s$). The online phase takes less than $1.7ms$, except the Citeseer2 and DBLP graphs and $k=32$ ($2.4ms$ and $2.7ms$ respectively). In conclusion, although \ReHub's  performance degrades for increasing values of $k$, its performance remains excellent throughout. Interestingly enough, for the top-4 graphs in terms of forward labels per vertex (Amazon, Citeseer 1 and 2, DBLP), the offline phase is $4.4 {-} 6 \times$ slower and the online phase $8{-}14 \times$ slower for $k=32$ than for $k=1$. This further demonstrates \emph{how the reduced size of the \RkNN\ backward labels (that depends on $k$) improves $ReHub$'s online phase performance}, in comparison to running a plain \emph{one-to-many} HL query between the query point and the objects (which would be much slower, even than the value observed for $k=32$). 

Regarding memory requirements, Figure~\ref{fig:memory2} reports the memory required for storing the additional data structures for \ReHub\ (\kNN\ backward labels, \kNN\ results per object and \RkNN\ backward labels) and the size of the \RkNN\ backward labels in comparison to the backward labels-to-many, with respect to varying values of $k$ (again for $D=0.01$). Results show that the memory required for the additional data structures for \ReHub\ is less than $50Mb$ even for $k=32$ and the worst performing graphs, whereas the size of the \RkNN\ backward labels can be more than $3 \times$ smaller than the size of the backward labels-to-many for the same value of $k$, i.e., the corresponding online phase will be consequently $3 \times$ faster than an one-to-many query, even for $k=32$. This shows that even for large values of $k$, $ReHub$'s online phase will still be very fast and efficient.   Moreover, our \emph{experimental results are entirely consistent with the theoretical analysis of $ReHub$ provided in Section~\ref{sub:complexity_analysis}} and show that the variable $\varepsilon$ introduced there, gets progressively smaller for larger values of $D$ and smaller values of $k$. 

\begin{figure}[!tb]
\centering
 \subfigure[Index size (Mb) for $ReHub$. Y-axis is on logarithmic scale] 
{\includegraphics[width=0.497\columnwidth]{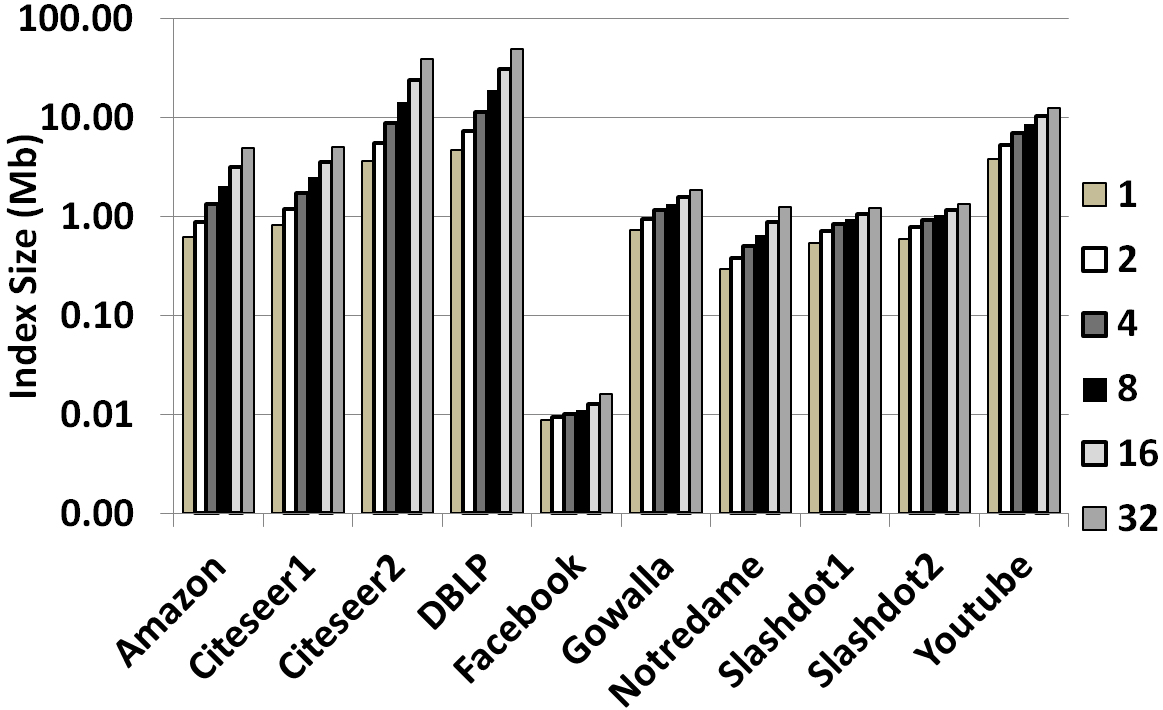}\label{fig:memory2a}} 
 \subfigure[Size of \RkNN\ backward labels in comparison to backward labels-to-many. Y-axis is on logarithmic scale]
{\includegraphics[width=0.497\columnwidth]{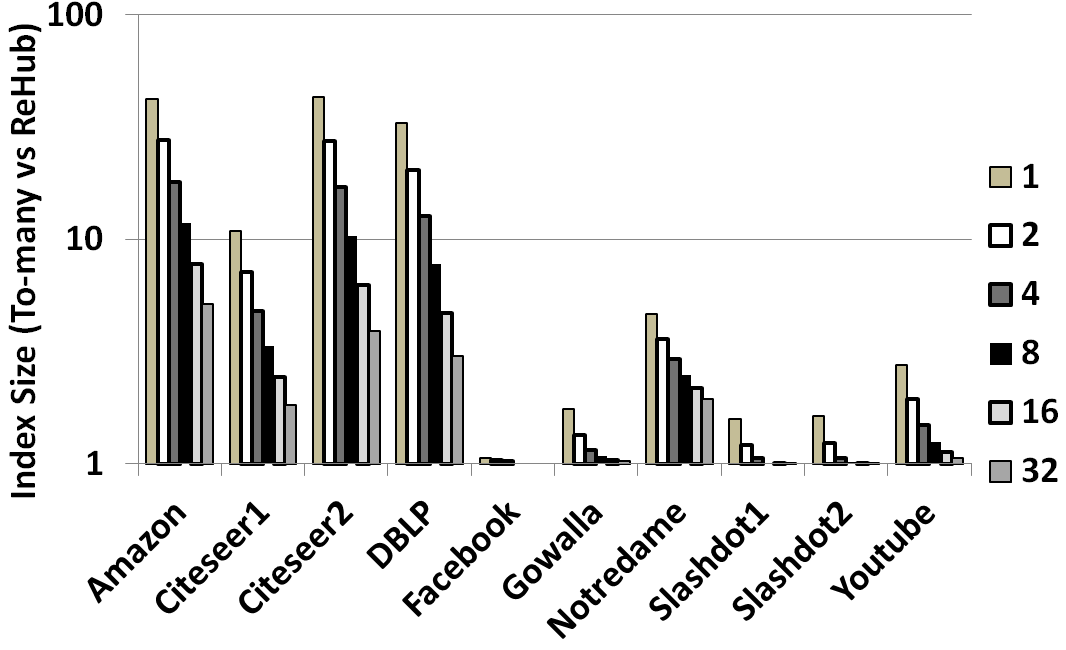}\label{fig:memory2b}}
 \caption{Memory footprint of \ReHub\ for $D=0.01$ and varying values of $k$}
 \label{fig:memory2}
\end{figure}  

On our third round of experiments, we evaluate the impact of objects distribution to ReHub's performance. To that purpose, we adapt a methodology similar to~\cite{delling2011g}. We pick a vertex at random and run BFS from it until reaching a predetermined number of vertices $|B|$. If $B$ is the set of vertices visited during this search, we pick our objects $O$ as a random subset of $B$. We keep the density of objects steady at $D=0.01$ and we experiment with different values of $|B|$ represented as percent of the total graph vertices. 

\begin{figure}[!tb]
\centering
 \subfigure[Offline phase. Y-axis is on logarithmic scale and time is reported in $ms$]
{\includegraphics[width=0.497\columnwidth]{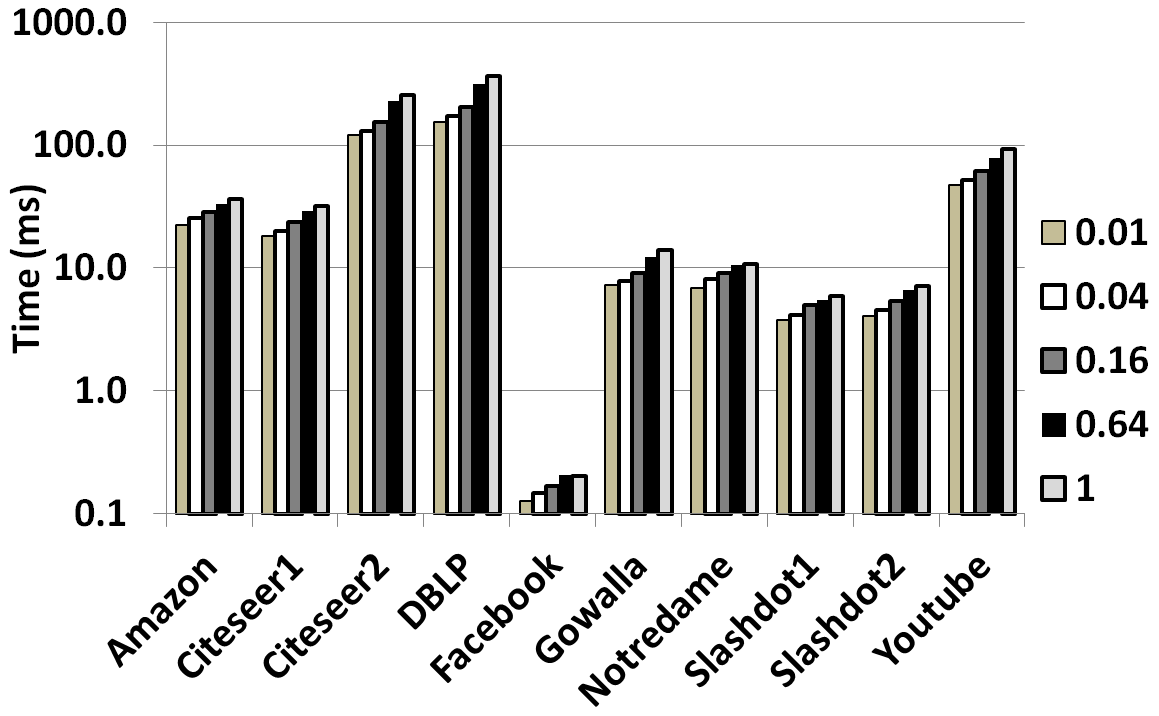}\label{fig:off3}} 
 \subfigure[Online phase. Y-axis is on logarithmic scale and time is reported in $\mu s$]
{\includegraphics[width=0.497\columnwidth]{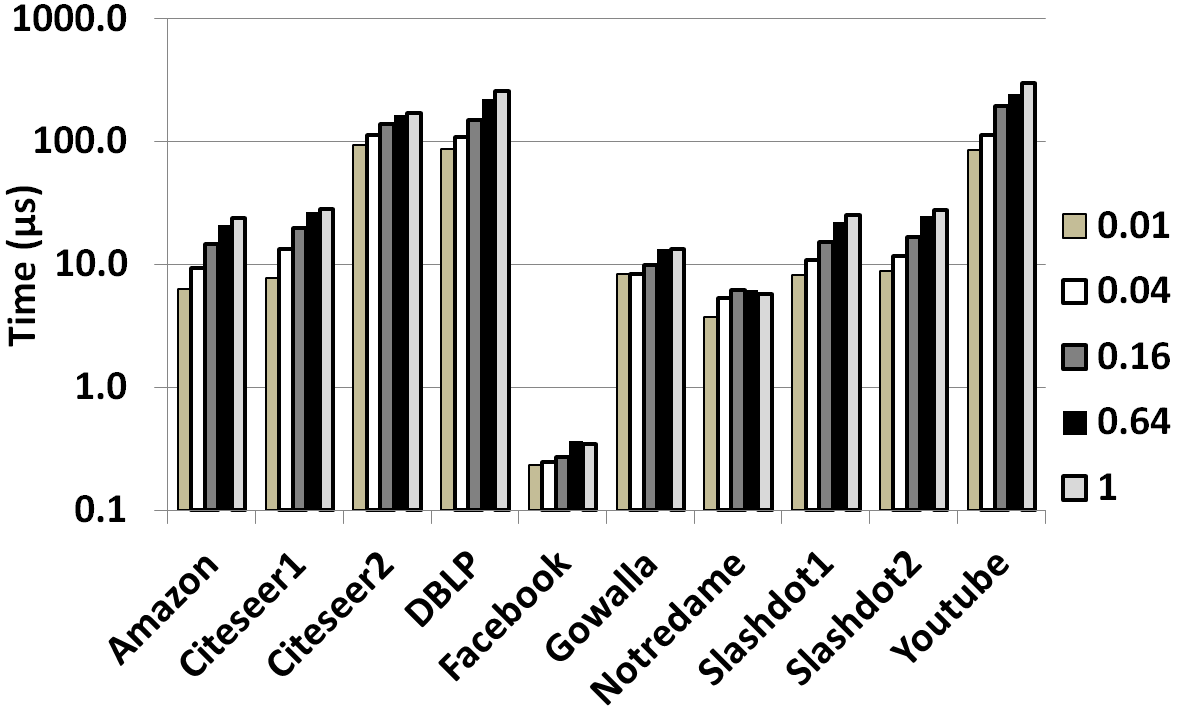}\label{fig:on3}}

 \caption{Offline and online phases of \ReHub\ for $k=1$, $D=0.01$  and varying values of $B$}
 \label{fig:k_b_0_01}
\end{figure} 

Again, we see that $ReHuB$ provides excellent performance both for the online and offline phase, regardless of the objects' distribution. In fact, ReHub performance is even better when the objects are more concentrated within the graph (e.g., for $B=0.01$) instead of randomly distributed objects ($B=1$). As a result, the offline phase is $1.6{-}2.3 \times$ faster for $B=0.01$ than $B=1$ and the online phase is $1.5{-}3.8 \times$  faster for $B=0.01$ than $B=1$. This great improvement, especially in the online phase, is attributed to the fact that the \RkNN\ backward labels are smaller, since objects are closer to each other, which facilitates faster online queries. This fact further testifies to the robustness of $ReHub$, even for more skewed distributions of objects. This is also evident in Figure~\ref{fig:memory3}  which shows that the memory requirements of $ReHub$ are smaller for smaller values of $B$ and the difference size between the RkNN\ backward labels in comparison to the backward labels-to-many is also amplified for smaller values of~$B$.

\begin{figure}[!tb]
\centering
 \subfigure[Index size (Mb) for $ReHub$. Y-axis is on logarithmic scale] 
{\includegraphics[width=0.497\columnwidth]{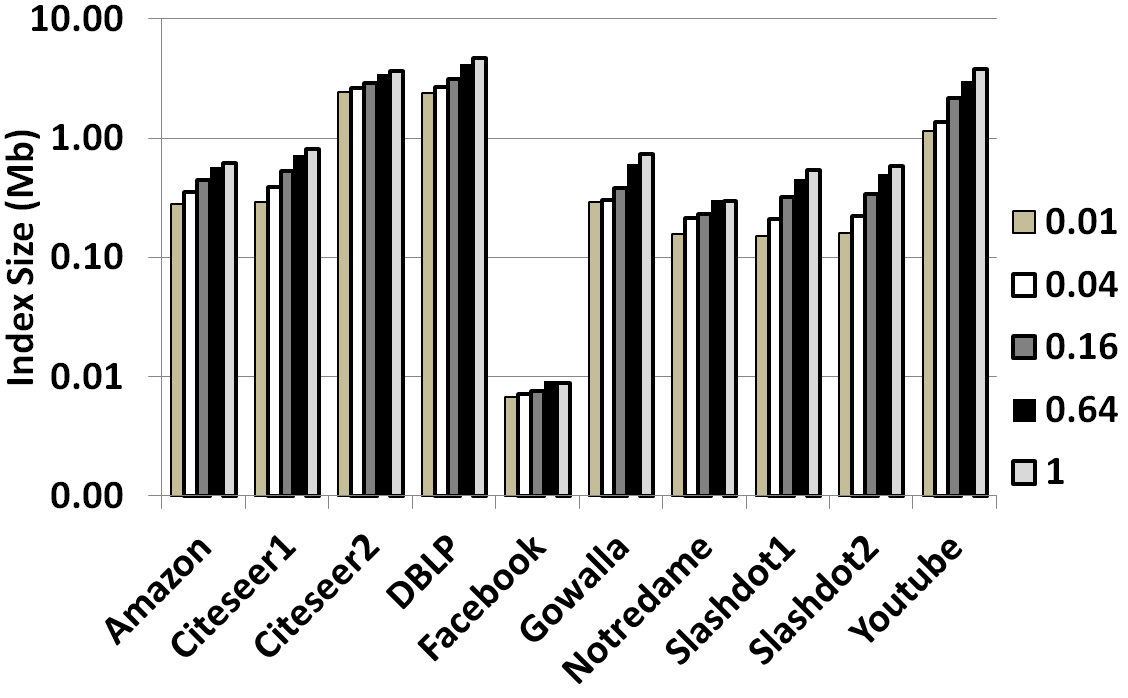}\label{fig:memory3a}} 
 \subfigure[Size of \RkNN\ backward labels in comparison to backward labels-to-many. Y-axis is on logarithmic scale]
{\includegraphics[width=0.497\columnwidth]{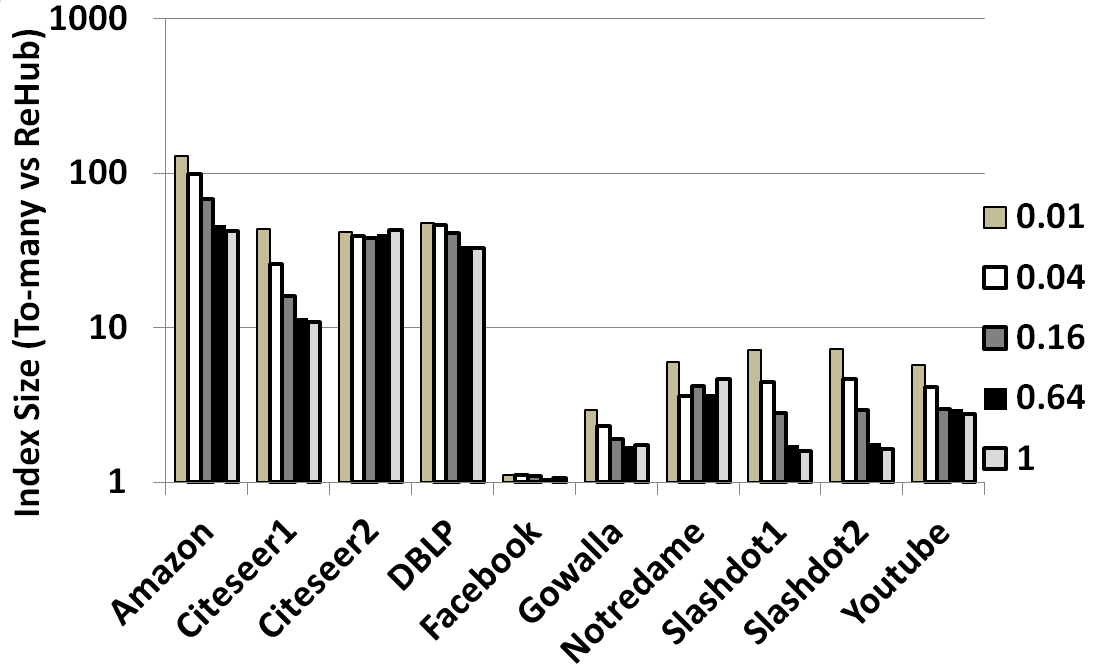}\label{fig:memory3b}}
 \caption{Memory footprint of \ReHub\ for $k=1$, $D=0.01$  and varying values of $B$}
 \label{fig:memory3}
\end{figure}


\subsection{Comparison with Previous works}
\label{sub:comparison} 

Our experimentation has shown that \ReHub\ exhibits excellent query performance and requires very small additional memory for all tested networks, regardless of the object density, the number~$k$ of \RkNN\ neighbors or the distribution of objects. In comparison to previous works, \ReHub\ \emph{may handle two orders of magnitude larger, denser networks than \cite{yiu2006,safar2009,borutta2014}, may scale easily for $k=32$}, where previous secondary storage methods have only been tested for up to $k=8$~\cite{yiu2006,borutta2014} or $k=20$~\cite{safar2009}. But even then,  e.g., for $k=8$ those methods require more than $300ms$~\cite{borutta2014}, whereas for similarly small networks (e.g. Gowalla) \ReHub's offline phase requires  $<20ms$ and the online phase $< 0.02ms$. Even for larger networks, the online phase typically requires less than $1ms$, i.e., \emph{ReHub is at least 3 orders of magnitude faster}. Thus,  \ReHub\ is the most complete solution for \RkNN\ queries on social and collaboration networks and the best overall contender for real-time applications.  In addition, Efentakis et al~\cite{efentakis2015b} have showed how the online phase of \ReHub\  may be translated to a simple SQL query on a open-source database engine, making \ReHub\ the only \RkNN\  solution that may also be used on a pure SQL context, for even greater versatility and scalability.    

Moreover, we showed that \ReHub\ can handle networks where the size of the created labels are more than three thousand hubs per vertex (e.g., DBLP, Citeseer2) and hence, the proposed algorithm will be even more efficient and faster when applied to sparser graph classes such as road networks, where the size of the created labels are less than a few hundred hubs per vertex, even for worse behaving metrics (e.g travel distances). 
\vspace{-5pt}

\pdfoutput=1

\section{Conclusion and Future Work}
\label{sec:conclusions}


This work introduced \ReHub, a novel main-memory algorithm that extends the hub labeling method to efficiently handle \RkNN\ queries on large-scale graphs. Our experimentation showed that \ReHub\ provides excellent query performance, has minimal memory requirements, and scales very well with the network size, the object density, the object distribution, the size of the labels, and the cardinality of the reverse $k$-nearest neighbor result. 
Given these results, \ReHub\ is the best overall and most complete solution for this type of queries and may also be used in the context of real-time applications. 

Directions for future work are to extend \ReHub\ towards handling object updates, i.e., objects may be added or deleted from the objects' set. Not having to redo the offline phase from scratch for such updates will significantly increase the practical applicability of the algorithm. Also testing our results on directed graphs and road networks will further showcase the algorithm's performance with respect to a wider range of graph classes, additional hub labelling algorithms, and domains.


\vspace{-5pt}
\bibliographystyle{abbrv}
\bibliography{gis2014}

\begin{thebibliography}{10}

\bibitem{abraham2012hldb}
I.~Abraham, D.~Delling, A.~Fiat, A.~V. Goldberg, and R.~F. Werneck.
\newblock Hldb: Location-based services in databases.
\newblock In {\em Proceedings of the 20th International Conference on Advances
  in Geographic Information Systems}, SIGSPATIAL '12, pages 339--348, New York,
  NY, USA, 2012. ACM.

\bibitem{abraham2011f}
I.~Abraham, D.~Delling, A.~Goldberg, and R.~Werneck.
\newblock A hub-based labeling algorithm for shortest paths in road networks.
\newblock In P.~Pardalos and S.~Rebennack, editors, {\em Experimental
  Algorithms}, volume 6630 of {\em Lecture Notes in Computer Science}, pages
  230--241. Springer Berlin Heidelberg, 2011.

\bibitem{abraham2012f}
I.~Abraham, D.~Delling, A.~Goldberg, and R.~Werneck.
\newblock Hierarchical hub labelings for shortest paths.
\newblock In L.~Epstein and P.~Ferragina, editors, {\em Algorithms – ESA
  2012}, volume 7501 of {\em Lecture Notes in Computer Science}, pages 24--35.
  Springer Berlin Heidelberg, 2012.

\bibitem{akiba2014f}
T.~Akiba, Y.~Iwata, K.~Kawarabayashi, and Y.~Kawata.
\newblock Fast shortest-path distance queries on road networks by pruned
  highway labeling.
\newblock In {\em 2014 Proceedings of the Sixteenth Workshop on Algorithm
  Engineering and Experiments, {ALENEX} 2014, Portland, Oregon, USA, January 5,
  2014}, pages 147--154, 2014.

\bibitem{akiba2013f}
T.~Akiba, Y.~Iwata, and Y.~Yoshida.
\newblock Fast exact shortest-path distance queries on large networks by pruned
  landmark labeling.
\newblock In {\em Proceedings of the {ACM} {SIGMOD} International Conference on
  Management of Data, {SIGMOD} 2013, New York, USA}, pages 349--360, 2013.

\bibitem{akiba2013fcode}
T.~Akiba, Y.~Iwata, and Y.~Yoshida.
\newblock Pruned landmark labeling [online].
\newblock \url{https://github.com/iwiwi/pruned-landmark-labeling}, 2015.

\bibitem{reka1999}
R.~Albert, H.~Jeong, and A.-L. Barabási.
\newblock The diameter of the world wide web.
\newblock {\em CoRR}, cond-mat/9907038, 1999.

\bibitem{dimacs2012}
D.~A. Bader, H.~Meyerhenke, P.~Sanders, and D.~Wagner, editors.
\newblock {\em Graph Partitioning and Graph Clustering - 10th DIMACS
  Implementation Challenge Workshop, Georgia Institute of Technology, Atlanta,
  GA, USA, February 13-14, 2012. Proceedings}, volume 588 of {\em Contemporary
  Mathematics}. American Mathematical Society, 2013.

\bibitem{bast2014}
H.~Bast, D.~Delling, A.~V. Goldberg, M.~M{\"{u}}ller{-}Hannemann, T.~Pajor,
  P.~Sanders, D.~Wagner, and R.~F. Werneck.
\newblock Route planning in transportation networks.
\newblock {\em CoRR}, abs/1504.05140, 2015.

\bibitem{borutta2014}
F.~Borutta, M.~A. Nascimento, J.~Niedermayer, and P.~Kr\"{o}ger.
\newblock Monochromatic rknn queries in time-dependent road networks.
\newblock In {\em Proceedings of the Third ACM SIGSPATIAL International
  Workshop on Mobile Geographic Information Systems}, MobiGIS '14, pages
  26--33, New York, NY, USA, 2014. ACM.

\bibitem{cheema2012}
M.~A. Cheema, W.~Zhang, X.~Lin, Y.~Zhang, and X.~Li.
\newblock Continuous reverse k nearest neighbors queries in euclidean space and
  in spatial networks.
\newblock {\em The VLDB Journal}, 21(1):69--95, Feb. 2012.

\bibitem{cho2011}
E.~Cho, S.~A. Myers, and J.~Leskovec.
\newblock Friendship and mobility: user movement in location-based social
  networks.
\newblock In {\em Proceedings of the 17th {ACM} {SIGKDD} International
  Conference on Knowledge Discovery and Data Mining, San Diego, CA, USA, August
  21-24, 2011}, pages 1082--1090, 2011.

\bibitem{cohen2002}
E.~Cohen, E.~Halperin, H.~Kaplan, and U.~Zwick.
\newblock Reachability and distance queries via 2-hop labels.
\newblock In {\em Proceedings of the Thirteenth Annual ACM-SIAM Symposium on
  Discrete Algorithms}, SODA '02, pages 937--946, Philadelphia, PA, USA, 2002.
  Society for Industrial and Applied Mathematics.

\bibitem{delling2011phast}
D.~Delling, A.~V. Goldberg, A.~Nowatzyk, and R.~F.~F. Werneck.
\newblock {PHAST:} hardware-accelerated shortest path trees.
\newblock In {\em 25th {IEEE} International Symposium on Parallel and
  Distributed Processing, {IPDPS} 2011, Anchorage, Alaska, USA, 16-20 May, 2011
  - Conference Proceedings}, pages 921--931, 2011.

\bibitem{delling2014f}
D.~Delling, A.~V. Goldberg, T.~Pajor, and R.~F. Werneck.
\newblock Robust distance queries on massive networks.
\newblock In {\em Algorithms - {ESA} 2014 - 22th Annual European Symposium,
  Wroclaw, Poland, September 8-10, 2014. Proceedings}, pages 321--333, 2014.

\bibitem{delling2013v}
D.~Delling, A.~V. Goldberg, and R.~F. Werneck.
\newblock Hub label compression.
\newblock In {\em Experimental Algorithms, 12th International Symposium, {SEA}
  2013, Rome, Italy, June 5-7, 2013. Proceedings}, pages 18--29, 2013.

\bibitem{delling2011g}
D.~Delling, A.~V. Goldberg, and R.~F.~F. Werneck.
\newblock Faster batched shortest paths in road networks.
\newblock In {\em {ATMOS} 2011 - 11th Workshop on Algorithmic Approaches for
  Transportation Modeling, Optimization, and Systems, Saarbr{\"{u}}cken,
  Germany, September 8, 2011}, pages 52--63, 2011.

\bibitem{delling2013h}
D.~Delling and R.~F. Werneck.
\newblock Customizable point-of-interest queries in road networks.
\newblock {\em {IEEE} Trans. Knowl. Data Eng.}, 27(3):686--698, 2015.

\bibitem{efentakis2015b}
A.~Efentakis, C.~Efstathiades, and D.~Pfoser.
\newblock {COLD}. {R}evisiting hub labels on the database for large-scale
  graphs.
\newblock In {\em Advances in Spatial and Temporal Databases - 14th
  International Symposium, {SSTD} 2015, Hong Kong, August 26-28. (To appear)},
  2015.

\bibitem{efentakis2014}
A.~Efentakis and D.~Pfoser.
\newblock {GRASP}. {E}xtending graph separators for the single-source
  shortest-path problem.
\newblock In A.~S. Schulz and D.~Wagner, editors, {\em Algorithms - ESA 2014},
  volume 8737 of {\em Lecture Notes in Computer Science}, pages 358--370.
  Springer Berlin Heidelberg, 2014.

\bibitem{efentakis2014c}
A.~Efentakis, D.~Pfoser, and Y.~Vassiliou.
\newblock {SALT}. {A} unified framework for all shortest-path query variants on
  road networks.
\newblock In E.~Bampis, editor, {\em Experimental Algorithms}, volume 9125 of
  {\em Lecture Notes in Computer Science}, pages 298--311. Springer
  International Publishing, 2015.

\bibitem{gavoille2001}
C.~Gavoille, D.~Peleg, S.~P{\'e}rennes, and R.~Raz.
\newblock Distance labeling in graphs.
\newblock In {\em Proceedings of the Twelfth Annual ACM-SIAM Symposium on
  Discrete Algorithms}, SODA '01, pages 210--219, Philadelphia, PA, USA, 2001.
  Society for Industrial and Applied Mathematics.

\bibitem{geisbergerss2008}
R.~Geisberger, P.~Sanders, and D.~Schultes.
\newblock Better approximation of betweenness centrality.
\newblock In J.~I. Munro and D.~Wagner, editors, {\em ALENEX}, pages 90--100.
  SIAM, 2008.

\bibitem{jiang2014}
M.~Jiang, A.~W. Fu, R.~C. Wong, and Y.~Xu.
\newblock Hop doubling label indexing for point-to-point distance querying on
  scale-free networks.
\newblock {\em {PVLDB}}, 7(12):1203--1214, 2014.

\bibitem{korn2000}
F.~Korn and S.~Muthukrishnan.
\newblock Influence sets based on reverse nearest neighbor queries.
\newblock In {\em Proceedings of the 2000 ACM SIGMOD International Conference
  on Management of Data}, SIGMOD '00, pages 201--212, New York, NY, USA, 2000.
  ACM.

\bibitem{snapnets}
J.~Leskovec and A.~Krevl.
\newblock {SNAP Datasets}: {Stanford} large network dataset collection.
\newblock \url{http://snap.stanford.edu/data}, June 2014.

\bibitem{leskovec2009}
J.~Leskovec, K.~J. Lang, A.~Dasgupta, and M.~W. Mahoney.
\newblock Community structure in large networks: Natural cluster sizes and the
  absence of large well-defined clusters.
\newblock {\em Internet Mathematics}, 6(1):29--123, 2009.

\bibitem{encyclopedia2009}
L.~Liu and M.~T. {\"O}zsu, editors.
\newblock {\em Encyclopedia of Database Systems}.
\newblock Springer US, 2009.

\bibitem{mcauley2012}
J.~J. McAuley and J.~Leskovec.
\newblock Learning to discover social circles in ego networks.
\newblock In {\em Advances in Neural Information Processing Systems 25: 26th
  Annual Conference on Neural Information Processing Systems 2012. Proceedings
  of a meeting held December 3-6, 2012, Lake Tahoe, Nevada, United States.},
  pages 548--556, 2012.

\bibitem{safar2009}
M.~Safar, D.~Ibrahimi, and D.~Taniar.
\newblock Voronoi-based reverse nearest neighbor query processing on spatial
  networks.
\newblock {\em Multimedia Systems}, 15(5):295--308, 2009.

\bibitem{yang2012b}
J.~Yang and J.~Leskovec.
\newblock Defining and evaluating network communities based on ground-truth.
\newblock In {\em 12th {IEEE} International Conference on Data Mining, {ICDM}
  2012, Brussels, Belgium, December 10-13, 2012}, pages 745--754, 2012.

\bibitem{yiu2006}
M.~L. Yiu, D.~Papadias, N.~Mamoulis, and Y.~Tao.
\newblock Reverse nearest neighbors in large graphs.
\newblock {\em Knowledge and Data Engineering, IEEE Transactions on},
  18(4):540--553, April 2006.

\end{thebibliography}



\end{document}